\documentclass[a4paper]{article}

\usepackage{amsthm,stmaryrd}
\usepackage{amsmath}
\usepackage{amssymb}
\usepackage{latexsym}

\usepackage{tcolorbox}

\usepackage{natbib}

\newtheorem{theorem}{Theorem}[section]
\newtheorem{proposition}[theorem]{Proposition}
\newtheorem{lemma}[theorem]{Lemma}
\newtheorem{corollary}[theorem]{Corollary}
\newtheorem{remark}[theorem]{Remark}
\newtheorem{definition}[theorem]{Definition}
\newtheorem{example}[theorem]{Example}
\newtheorem{assumption}[theorem]{Assumption}

\newcommand{\bW}{\mathbf{W}}
\newcommand{\bw}{\mathbf{w}}
\newcommand{\bC}{\mathbf{C}}
\newcommand{\bL}{\mathbf{L}}
\newcommand{\bK}{\mathbf{K}}
\newcommand{\bV}{\mathbf{V}}
\newcommand{\bff}{\mathbf{f}}
\newcommand{\bS}{\mathbf{S}}
\newcommand{\bQ}{\mathbf{Q}}

\newcommand{\bZ}{\mathbf{Z}}
\newcommand{\bh}{\mathbf{h}}

\def\A{{\mathcal{A}}}
\def\H{{\mathcal{H}}}
\def\F{{\mathcal{F}}}

\begin{document}

\bibliographystyle{apalike}

\title{Super-replication with transaction costs under model uncertainty for continuous processes 
\footnote{Huy N. Chau is supported by Center for Mathematical Modeling and Data Science, Osaka University. 
Masaaki Fukasawa is supported by Osaka University. 
Mikl\'os R\'asonyi is supported by the ``Lend\"ulet'' grant LP 2015-6 of the Hungarian Academy of Sciences.} }
\author{Huy N. Chau \and Masaaki Fukasawa \and Mikl\'os R\'asonyi}

\maketitle

\begin{abstract}
We formulate a superhedging theorem in the presence of transaction costs and model uncertainty.  
Asset prices are assumed continuous and uncertainty is modelled in a parametric setting. Our proof
relies on a new topological framework in which no Krein--Smulian type theorem is available. 
\end{abstract}

\section{Introduction}

Robust super-replication has a rich literature, starting from the seminal work of \cite{hobson1998} about lookback options. 
In this paper, bounds for option prices and the connection with the Skorokhod embedding are established. The approach was employed further by \cite{brown} for barrier options, \cite{cox2011} for double no-touch options,  \cite{hobson2012} for forward start options, \cite{carr} for variance options, among others. We refer to the survey of \cite{hobson2011} and references therein. 

The use of optimal transport in robust hedging was developed in \cite{galichon2014}, \cite{tan2013}, \cite{bhp2013}, 
\cite{soner}. These works provided a dual formulation which transforms the superhedging problem into a martingale optimal 
transportation problem. In a continuous time setting, \cite{soner} proved duality and the existence of a family of simple, 
piecewise constant
super-replication strategies that asymptotically achieve the minimal super-replication cost. In the quasi-sure 
setting, \cite{denis2006} established a theoretical framework to construct stochastic integral and duality results, 
while \cite{nutz}, \cite{bn15} obtained optimal super-replication strategies in discrete time.  Various 
frameworks are investigated in \cite{davis2014, maggis, jan, michael,dito, acciaio2016model}.

Under transaction costs, we refer to the discrete-time studies \cite{bouchard,dolinsky-soner1}.  As pointed out by \cite{bouchard},
it is not easy to come up with a proof of the superhedging duality using the local method, developed in \cite{bn15}, which 
uses one-period argument and then measurable selection techniques to paste the periods together. In \cite{bouchard}, 
the authors introduced a fictitious market without transaction cost using a randomization technique, then the results of \cite{bn15} were applied.  In continuous-time markets with transaction costs, however, the only results we know of are those of \cite{dolinsky-soner2}.
They establish, under technical conditions, that the pathwise superhedging price is the same as the superhedging price
for a continuous price process satisfying the conditional full support property. 

The study \cite{biagini2015robust} is close in spirit to our paper hence a brief comparison 
is in place. In \cite{biagini2015robust} arbitrage in the quasi-sure
sense was analysed in \emph{frictionless} markets with continuous price processes, 
using the ``no arbitrage of the first kind'' assumption $NA_{1}(\mathcal{P})$, where $\mathcal{P}$ denotes
the set of possible prior probabilities. A superhedging theorem is provided under this hypothesis. 
A crucial step in their arguments shows that $NA_{1}(\mathcal{P})$ 
is equivalent to having $NA_{1}(P)$ for all individual $P\in\mathcal{P}$.
In the present paper, somewhat analogously, we assume a ``no free lunch with vanishing risk'' condition
(stronger than no arbitrage of the first kind) 
for each individual model and deduce the superhedging theorem (with transaction costs) under this assumption.
Note that we operate with ``martingale'' price systems
while \cite{biagini2015robust} uses ``supermartingale deflators''.

In the present article we prove a fairly general superhedging theorem. 
We take a parametrization approach, different from the pathwise or quasi-sure settings,
first used in \cite{chau2019robust}, \cite{rasonyi2018utility}, see also \cite{chau2020}. 
As far as we know, 
the present paper is the first to apply 
to a wide range of continuous-time markets under transaction costs and model uncertainty. We now list the main features of the present paper.

First, we use the new topological framework of \cite{chau2020} for studying hedging under model 
uncertainty, this time in a continuous-time setting. The working space is the 
product topological vector space $\bL^{\infty} = \prod_{\theta \in \Theta} (L^{\infty}, w^*)$ 
where we denote the product topology by $\bw^*$. Here $\Theta$ is the set of parameters describing uncertainty and $w^*$ is the 
weak star topology on $L^{\infty}$, the space of a.s.\ equivalence classes of bounded random variables. 
Unlike the local method of \cite{bn15}, this approach offers a global method, 
which is suitable for handling transaction costs in continuous time, 



Second, for the separation arguments to work, we need the $\bw^*$-closedness of 
the set of hedgeable payoffs $\bC$ in the product space $\bL^{\infty}$. This is typically shown 
relying on the Krein--Smulian theorem. Usual versions of that result are stated 
for Fr\'echet spaces, see for example \cite{schachermayer1994martingale}. Our product space 
$\bL^{\infty}$ has the predual $\bigoplus_{\theta \in \Theta} L^{1}$, which is, in general, \emph{not} Fr\'echet, 
as $\Theta$ is typically uncountable. Therefore we are not able to apply the Krein--Smulian theorem directly. 
To remedy this, we prove a convex compactness property for the set of strategies in a suitable topology.
Such a property fails in frictionless markets, it is a particularity of markets with transaction costs. 
Next we apply Krein--Smulian for finite direct sums of Fr\'echet spaces and the $\bw^*$-closedness of $\bC${}
will be ensured by convex compactness.

Third, we apply the Hahn-Banach theorem in $\bL^{\infty}$ to obtain what will be called consistent price 
systems in the robust sense. 
These systems are, in fact, infinite dimensional vectors with finitely many nonzero components such that  
discounting price processes by them has the same effect as the usual consistent price systems.
  
The paper is organized as follows. In Section \ref{sec:model}
we introduce the market model; consistent
price systems in the robust sense; and the main result. 
Proofs are given in Section \ref{sec:proof}. Section \ref{sec:app} recalls important facts about
topological vector spaces and establishes the necessary results about
convex compactness.
Extensions of our results to the multi-asset, conic framework of \cite{kabanov2009markets} seem
straightforward but they would involve technical complications so we are staying in a one-asset framework for now.

It seems less obvious to
include stock prices with jumps where trading strategies 
are not necessarily right continuous, see \cite{guasoni2012fundamental}. This would
necessitate finding a topological space in which the set of nice trading strategies is convexly compact. 
It would also be interesting to see if the present framework can be adapted to continuous-time frictionless markets,
see \cite{chau2020} about what happens in the discrete-time case.

\textit{Notations}. Let $I$ be some index set and $X_{i}$, $i\in I$ be sets.
In the product space  $\mathbf{X}=\prod_{i \in I} X_i$, a vector $(f^i)_{i \in I}$ will be 
denoted by $\mathbf{f}$. If there are an orderings $\geq_{i}$ given on each $X_{i}$ then
we write $\mathbf{f} \ge \mathbf{g}$ if $f^i \geq_{i} g^i$ for all $i \in I$. 
If $1\in X_{i}$ for all $i$ then $\mathbf{1}$ denotes the vector with all coordinates equal to $1$ 
and $\mathbf{1}^{i}$ denotes the vector with coordinate $i$ equal to $1$ and the other coordinates zero.  
Similarly, when $0\in X_{i}$, $\mathbf{0}$ denotes the vector all of whose coordinates equal $0$. 
The vector space of (equivalence classes of) random variables on $(\Omega, \mathcal{F}, P)$ is denoted by $L^0(\mathcal{F}, P)$. 
As usual, $L^{p}(\mathcal{F}, P), p \in [1, \infty]$ is the space of $p$-integrable (resp. bounded)
random variables equipped with the 
standard $\|\cdot\|_{p}$ norm.

\section{The model}\label{sec:model}

Let $T>0$ be the time horizon, let $(\Omega, \mathcal{F}, (\mathcal{F}_t)_{t\in [0,T]}, P)$ be a filtered probability 
space, where the filtration is assumed to be right-continuous and $\mathcal{F}_0$ 
coincides with the $P$-completion of the trivial sigma-algebra. For simplicity, we also
assume $\mathcal{F}=\mathcal{F}_{T}$. Consider a financial 
market model with one risky asset and one risk-free asset (cash) whose price is assumed to
be constant one. 

Let $\Theta$ be a (non-empty) set, which is interpreted as the parametrization of uncertainty. 
It is \emph{not} required that $\Theta$ is a convex subset of a vector space.  
Consider a family $(S^{\theta}_t)_{t\in [0,T]}$, 
$\theta \in \Theta$ of adapted, positive processes with continuous trajectories
which represent the possible evolutions of the price of the risky asset in consideration. 

The risky asset is traded under proportional transaction cost $\lambda \in (0,1)$, where 
the possible ask prices and bid prices are $(1+\lambda)S^{\theta}_t$,  $(1-\lambda)S^{\theta}_t$, respectively. 


Let $H^{\uparrow}_{t}$, $H^{\downarrow}_{t}$, $t\in [0,T]$ be predictable processes with non-decreasing left-continuous
trajectories. $H^{\uparrow}_{t}$ denotes the cumulative amount of transfers from the riskless asset to the risky one 
up to time $t$ and 
$H^{\downarrow}_{t}$ represents the cumulative transfers in the opposite direction. 
The set of all such $(H^{\uparrow},H^{\downarrow})$ is denoted by $\mathbf{V}$. 
A useful metric structure can be defined on $\mathbf{V}$, see Subsection \ref{section:v_space}, which 
will be key for later developments. 

The portfolio position in
the risky asset at time $t$ equals 
$$
\phi_t:=H_0+H^{\uparrow}_t - H^{\downarrow}_t,\ t\in [0,T].
$$ 
Here $H_{0}$ is a ($\mathcal{F}_{0}$-measurable, hence deterministic) initial transfer.
We also assume that $\phi_{0-}:=0$.

For any real number
$x\in\mathbb{R}$, we denote $x^+:=\max\{0,x\}$, $x^-:=\max\{0,-x\}$.
For a pair of initial capital and initial transfer $(z, H_0) \in \mathbb{R}^2$ and a strategy $H \in \bV$, 
the corresponding liquidation value for parameter $\theta\in\theta$ at time $t\in [0,T]$ is defined by 
\begin{eqnarray}
W^{z}_t(S^{\theta}, H, \lambda) &:=&  z - H_0^+S^{\theta}_0(1+\lambda)+ H_0^-S^{\theta}_0(1-\lambda) - 
\int_0^t{(1+\lambda)S^{\theta}_u dH^{\uparrow}_u} \nonumber \\
&+& 
\int_0^t{(1-\lambda)S^{\theta}_udH^{\downarrow}_u} + \phi^{+}_t(1-\lambda)S^{\theta}_t - \phi^-_t(1+\lambda)S^{\theta}_t. \label{eq:liq} 
\end{eqnarray}


\begin{definition}\label{defi:admi_appro}
A finite variation process $(H_{0},H) \in\mathbb{R}\times\mathbf{V}$ is called an $x$-admissible strategy 
	for the model $\theta$ if 
\begin{equation}\label{}
	W^{0}_t(S^{\theta}, H, \lambda)  \ge - x,\mbox{ a.s. for all }t\in [0,T].
	\end{equation}
	Denote by $\mathcal{A}^{\theta}_{x}(\lambda)$ the set of $x$-admissible strategies for the model $\theta$,{}
	set $\mathcal{A}^{\theta}(\lambda)= \bigcup_{x> 0}\mathcal{A}^{\theta}_{x}(\lambda)$. 
	The set of admissible strategies for the robust model is defined by 
	$\mathcal{A}(\lambda) := \bigcap_{\theta \in \Theta} A^{\theta}(\lambda)$.
	{}
	
\end{definition}



Consider the product spaces $\bL^0_t = \prod_{\theta \in \Theta} L^0(\mathcal{F}_t,P), \bL^{\infty}_t = \prod_{\theta \in \Theta} L^{\infty}(\mathcal{F}_t,P)$ with the corresponding product topologies. Other product spaces are defined similarly. For notational convenience, we use bold notations for vectors in these product spaces. For $t \in [0,T]$, we denote by $\bS_t := (S^{\theta}_t)_{\theta \in \Theta}$ a vector in $\bL^0_t$. For $H \in  \mathcal{A}(\lambda)$, we denote by
$$\bW^0_t(H) = (W^0_t(\theta, H, \lambda))_{\theta \in \Theta}$$
the vector consists of attainable payoffs from the strategy $H$ under model uncertainty. 
Set $\mathbf{L}^{0}:=\mathbf{L}_{T}^{0}$, $\mathbf{L}^{\infty}:=\mathbf{L}^{\infty}_{T}$. 
$\mathbf{L}_{+}^{0}$ is the non-negative orthant of $\mathbf{L}^{0}$. Define
$$\bK_0 = \left\lbrace \bW^0_T(H): H \in \A(\lambda) \right\rbrace, \qquad \bC = (\bK_0 - \bL^0_+) \cap \bL^{\infty}. $$

For each $\theta \in \Theta$, the classical no free lunch with vanishing risk condition for $S^{\theta}$, denoted by $NFLVR(\lambda, \theta)$, is recalled from \cite{guasoni2012fundamental}. 

\begin{definition}
The price process $S^{\theta}$ together with the transaction cost $\lambda$ satisfy the $NFLVR(\lambda, \theta)$ condition if for any sequence $H_n, n \in \mathbb{N}$ such that $H_n \in \mathcal{A}^{\theta}_{1/n}$, and $W^0_T(\theta,H_n, \lambda)$ converges a.s. to some limit $W \in [0,\infty]$ a.s. then $W = 0, a.s.$ 
\end{definition}

The following conditions will be imposed throughout the paper:

\begin{assumption}\label{assumption_main}
\begin{itemize}
\item[(i)] There is a sequence of stopping times $T_n, n \in \mathbb{N}$ increasing almost surely to infinity such that $\sup_{\theta \in \Theta}\sup_{t\in [0,T]}S^{\theta}_{t \wedge T_n} \le K_n$ a.s.\ for some constants $K_{n}$, $n\in\mathbb{N}$.
\item[(ii)] For all $\theta \in \Theta,$ 
the stock price $S^{\theta}$ satisfies $NFLVR(\lambda, \theta)$ for all $0<\lambda < 1$. 
\end{itemize}
\end{assumption}

Condition $(i)$ is technical, requiring uniform boundedness of the prices up to a sequence of stopping times. 
Condition $(ii)$ assumes that none of the possible price processes admits
(model-dependent) arbitrage. It holds obviously that 
$\bC \cap \bL^{\infty}_+ = \{ \mathbf{0} \}$. 
   
\begin{remark}{\rm To develop an intuition about the current, parametric framework of uncertainty
we sketch a rather general example. Let $Y$ be some stochastic process with values in $\mathbb{R}^{m}$
for some $m$ which generates the filtration. We imagine that $Y$ describes 
economics factors (dividend yields, interest rates, unemployment rates, statements of firms)
which are known (at least we have a reliable statistical model for them). The asset price $S$
is assumed to be a nonlinear functional of these factors, that is, $S^{\theta}_{t}=F(\theta,Y_{t})$, $\theta\in\Theta$ for
some parameter set $\Theta$ and a function $F:\Theta\times\mathbb{R}^{m}\to\mathbb{R}_{+}$.
Here the parameter $\theta$ is \emph{unknown}. 
Note that if $Y$ is locally bounded and $F$ is regular enough then a localizing sequence $T_{n}$ for $Y$ satisfies (i) in Assumption \ref{assumption_main}.
}
\end{remark}


\begin{example}
{\rm Let the filtration be generated by a Brownian motion $W_{t}$.
Consider the possible price processes
$$
S^{\theta}_{t}= F \left(\int_{0}^{t}\mu^{\theta}(s)\, ds+\int_{0}^{t}K^{\theta}(t,s)dW_{s}\right) ,  \ t\in [0,T]
$$
where $\theta$ runs in an index set $\Theta$, $K^{\theta}(t,s)$
are kernels of a suitable regularity, $\mu^{\theta}$ are suitable processes and $F: \mathbb{R} \to (0,\infty)$ is a 
function (e.g.\ exponential). This is a family of
models which possibly fail the semimartingale property but satisfy $NFLVR(\lambda,\theta)$ for every
$0<\lambda<1$ under mild conditions, see e.g.\ \cite{grw}, \cite{gsv}.}
\end{example}

\begin{definition}\label{def:pricing_func}
A consistent price system in the robust sense 
for $\bS$  is a pair $(\bZ, M)$ where
\begin{itemize}
	\item[(i)] $0 \le \mathbf{Z}_T \in \bigoplus_{\theta \in \Theta} L^1$, and $E[Z^{\bar{\theta}}_T1_A] > 0$ for some $\bar{\theta}$ and for some 	$A\in\mathcal{F}$. 
	\item[(ii)] $M$ is a local martingale such that 
	\begin{equation}\label{eq:sandwich}
	\sum_{\theta \in \Theta} (1-\lambda) Z^{\theta}_tS^{\theta}_t \le M_t \le 
	\sum_{\theta \in \Theta} (1+\lambda) Z^{\theta}_tS^{\theta}_t, a.s.,
	\end{equation}
	where $Z^{\theta}_t:= E[Z^{\theta}_T|\mathcal{F}_t], t \in [0,T]$ for all $\theta\in\Theta$.
\end{itemize}
The set of such objects is denoted by $\mathcal{C}(\bar{\theta}, \lambda, A)$. Let 
$$
\mathcal{Z}(\lambda):=\bigcup_{A \in \mathcal{F},\bar{\theta}\in\Theta}\mathcal{C}(\bar{\theta}, \lambda, A).
$$
\end{definition}

Note that $\mathcal{C}(\bar{\theta}, \lambda, A)$ is empty when $P(A)=0$ so we could have 
restricted the above definition to sets $A$ with $P(A)>0$.
The quantity $\bZ_T$ may be interpreted as a local martingale density. 
In the setting without uncertainty, it is usually required that $Z^{\bar\theta}_T >0, a.s.$ 
Here, our density $Z^{\bar{\theta}}_T$ has to be positive only on some $A$ with $P[A] >0$.
The new densities might look strange at the first glance, however, they are comparable to 
``absolutely continuous local martingale measures", introduced in \cite{ds}, \cite{ds2}. 
The main reason for 
using this notion is that, $\Theta$ being possibly uncountable,
no exhaustion argument (as in the Kreps-Yan separation theorem) can be applied.
It will become clear in the proof of Theorem \ref{super} 
that the new densities work better than the densities with $Z^{\bar{\theta}}_T >0, a.s.$. 

When $\Theta$ is a singleton, it is possible to use the exhaustion argument and we may assume 
$Z^{\bar{\theta}}_T >0, a.s.$ so this 
definition reduces to the usual definition of consistent price systems, see 
\cite{grw}, \cite{guasoni2012fundamental}. 

\begin{remark}\label{kabi}
{\rm Let $D\subset\Theta$ be finite and let $\mathbf{Z}^{i}_{T}=0$ for all the coordinates $i\notin D$. 
If, for each $\theta\in D$, 
\begin{equation}\label{eq:sandwich1}
	(1-\lambda) Z^{\theta}_t S^{\theta}_t \le M^{\theta}_t \le 
(1+\lambda) Z^{\theta}_tS^{\theta}_t, a.s.
\end{equation}
for suitable martingales $M^{\theta}$ (that is, the processes $(Z^{\theta},M^{\theta})$ determine
consistent price systems for the individual models $S^{\theta}$, $\theta\in D$) then,
defining $M_t:=\sum_{\theta\in D}M^{\theta}_t$ defines a martingale satisfying 
\eqref{eq:sandwich}. This shows how ``usual'' consistent price systems provide
objects in $\mathcal{C}(\bar{\theta}, \lambda, A)$. The important point is
that there are many \emph{other} elements in $\mathcal{C}(\bar{\theta}, \lambda, A)$.
This is explained in Subsection 3.5 of \cite{chau2020} where
frictionless toy examples are presented.}
\end{remark}



We state the main result of the paper.
\begin{theorem}\label{super}
Let Assumption \ref{assumption_main} be in force. 
Let $\mathbf{G} = (G^{\theta})_{\theta \in \Theta}\in\mathbf{L}^{0}$ be bounded or 
let $\mathbf{G} \ge 0$. There exists $(z,H)\in \mathbb{R} \times \mathcal{A}(\lambda)$ such that
$$
W^{z}_{T}(S^{\theta},H,\lambda)\geq G^{\theta},\ \theta\in\Theta
$$
if and only if
$$
z\geq\sup_{(\bZ,M)\in \mathcal{Z}(\lambda)}E\left[\sum_{\theta\in\Theta}Z_{T}^{\theta} G^{\theta}\right] .
$$	
\end{theorem}

The proof of Theorem \ref{super} is given in Section \ref{sec:proof} below.

\section{Proof}\label{sec:proof}

\subsection{The $\bw^*$-closedness of $\bC$}
The following result is classical.
\begin{lemma}\label{lem:bounded_P}
For any $x>0$, and any $\theta \in \Theta$, the set $\{ \|H\|_T: H \in \mathcal{A}^{\theta}_x(\lambda) \} $ is bounded in probability.
\end{lemma}
\begin{proof}
We can assume that $x=1$ and adopt the argument in \cite{guasoni2012fundamental}. By Assumption \ref{assumption_main} (ii), for each $\theta \in \Theta$, $0 < \widehat{\lambda} < \lambda$, the process $S^{\theta}$ satisfies the classical condition $NFLVR(\widehat{\lambda}, \theta)$. Let us denote
\begin{eqnarray*}
\overline{S}^{\theta} = (\lambda - \widehat{\lambda}) \inf_{t \in [0,T]}  S^{\theta}_t > 0, a.s.
\end{eqnarray*}
Assume that there exist $\alpha > 0$ and a sequence $H^n \in \mathcal{A}^{\theta}_1(\lambda), n \in \mathbb{N}$ such that for every $n \in \mathbb{N}$,
$$P\left[ \overline{S}^{\theta}\| H^n\|_T > n \right] \ge \alpha.$$
By Definition \ref{defi:admi_appro}, we have that 
\begin{eqnarray*}
&&1  + W^0_t(S^{\theta},H^n, \widehat{\lambda})\\
&\ge& W^0_t(S^{\theta},H^n, \widehat{\lambda}) - W^0_t(S^{\theta},H^n, \lambda)\\
&=&   \int_{0}^t{\left( (1+\lambda) S^{\theta}_u - (1+ \widehat{\lambda}) S^{\theta}_u \right) dH^{n,\uparrow}_u} \\
&+&\int_{0}^t{\left( (1- \widehat{\lambda}) S^{\theta}_u - (1-\lambda) S^{\theta}_u  \right) dH^{n,\downarrow}_u} \\
&+& (H^{n}_t)^+\left( (1- \widehat{\lambda}) S^{\theta}- (1-\lambda) S^{\theta}_u  \right) + 
(H^{n}_t)^- \left( (1+\lambda) S^{\theta}_t - (1+ \widehat{\lambda}) S^{\theta}_t \right)  \\
&\ge&   \overline{S}^{\theta} \|H^n\|_t.
\end{eqnarray*}
It is clear that $\tilde{H}^n = H^n/n \in \mathcal{A}_1^{\theta}(\lambda)$ and  then
\begin{equation}\label{eq:tv}
1/n + W^0_T(S^{\theta},\tilde{H}^n, \widehat{\lambda}) \ge  \overline{S}^{\theta} \|\tilde{H}^n \|_T, a.s.
\end{equation}
and $P\left[  \overline{S}^{\theta} \|\tilde{H}^n\|_T > 1\right] \ge \alpha$ for all $n \in \mathbb{N}$. By Lemma 9.8.1 of \cite{delbaen2006mathematics}, there exist convex combinations $f^n \in \text{conv}\{\overline{S}^{\theta}\| \tilde{H}^n\|_T, \overline{S}^{\theta}\|\tilde{H}^{n+1}\|_T,...\}$  such that $f^n \to f, a.s.$ and 
\begin{equation}\label{eq:pos}
P[f > 0] > 0.
\end{equation}
Using the same weights as in the construction of $f^n$, we obtain a sequence of strategies $\widehat{H}^n$ such that
$$1/n + W^0_T(S^{\theta},\widehat{H}^n, \widehat{\lambda}) \ge f^n, a.s..$$
Therefore, the random variable $f$ 
is a $FLVR(\widehat{\lambda}, \theta)$ for $S^{\theta}$, which is a contradiction.
\end{proof}

\begin{remark}
{\rm If $H$ is a convex combination of $H^1,H^2$, 
it may happen that there are overlapping regions 
where $H^{\uparrow}_t H^{\downarrow}_t \ne 0$, for example 
when $H^{2,\uparrow}_tH^{1,\downarrow}_t \ne 0.$ 
However, we can always remove these redundant transactions from $H$ and 
the obtained strategy generates a payoff which is at least as that from $H$. Without further notice, 
from now on we interpret $H$ as the strategy after removing redundancy.}   
\end{remark}

Let $D$ be a non-empty finite subset of $\Theta$. We say that a sequence $\bff_n, n \in \mathbb{N}$ in $\prod_{\theta \in D}L^0$ Fatou-converges to $\bff$ if for each $\theta \in D$, $f^{\theta}_n$ converges to $f^{\theta}$ a.s. and $f^{\theta}_n \ge -x^{\theta}, a.s.$ for some $x^{\theta} >0$. Define
$$\bC(D) = \left\lbrace (W^0_{T}(S^{\theta}, H, \lambda) - h^{\theta})_{\theta \in D}: H \in \bigcap_{\theta \in D} \mathcal{A}^{\theta}(\lambda), h^{\theta} \in L^0_+, \theta \in D\right\rbrace \bigcap \prod_{\theta \in D}L^{\infty}(\mathcal{F}_T,P).$$
\begin{proposition}\label{pro:dense_closed}
The set $\bC(D)$ is Fatou-closed, that is $\bC(D) \cap \prod_{\theta \in D} B^{\infty}_{x^{\theta}}$ is closed in the product space $\prod_{\theta \in D} L^0(\mathcal{F}_T,P)$, for each $(x^{\theta})_{\theta \in D} \in \mathbb{R}^{D}_+$.
\end{proposition}
\begin{proof}
Let $\bff_n, n \in \mathbb{N}$ be a sequence in $\bC(D)$ such that for every $\theta \in D$, $f^{\theta}_n \to f^{\theta}$ in probability and $f^{\theta}_n \ge - x^{\theta}$.  We need to find $H\in \bigcap_{\theta \in D} \mathcal{A}^{\theta}(\lambda)$ such that for every $\theta \in D,$ $W^0_{T}(S^{\theta},H,\lambda) \ge f^{\theta}, a.s.$.

By taking a subsequence, we may assume that $f^{\theta}_n \to f^{\theta}, a.s.$ for all $\theta \in D$. By definition, 
$$f^{\theta}_n \le W^0_T(S^{\theta},H^n,\lambda), \forall \theta \in D$$ for some $H^n \in \bigcap_{\theta \in D} \mathcal{A}^{\theta}(\lambda)$. 
Using Theorem 1 of \cite{schachermayer2014admissible}, for every $\theta \in D$ we get that $H^n \in \mathcal{A}^{\theta}_{x^{\theta}}(\lambda)$ for all $n \in \mathbb{N}$.
Fix $\theta_0 \in D$ arbitrarily. Lemma \ref{lem:bounded_P} implies that the set $\{\|H\|_T, H \in \mathcal{A}^{\theta_0}_{x^{\theta_0}}(\lambda)\}$ is bounded in probability.  Using Lemma B.4 of \cite{guasoni2012fundamental}, there are convex combinations $\tilde{H}^{n,\uparrow} \in conv(H^{n,s,\uparrow}, H^{n+1,\uparrow},...)$ and a finite variation process $H^{\uparrow}$ such that $\tilde{H}^{n,\uparrow} \to H^{\uparrow}$ pointwise and $\|\tilde{H}^{n,\uparrow}\| \to \|H^{\uparrow}\|$ pointwise. Similarly, we find another convex combinations, still denoted by $\tilde{H}^{n,\downarrow}$, such that $\tilde{H}^{n,\downarrow} \to H^{\downarrow}$ pointwise and $\|\tilde{H}^{n,\downarrow}\| \to \|H^{\downarrow}\|$ pointwise. 

Next we prove that $H \in \bigcap_{\theta \in D} \mathcal{A}^{\theta}(\lambda)$. First, using Lemma 4.3 of \cite{guasoni2002optimal}, we obtain that 
$\tilde{H}^{n,\uparrow} \to H^{\uparrow}$ and  $\tilde{H}^{n,\downarrow} \to H^{\downarrow}$ weakly. It follows that for every $\theta \in \Theta, t \in [0,T]$,
$$W^0_t(S^{\theta},\tilde{H}^{n}, \lambda) \to W^0_t(S^{\theta},H, \lambda), a.s..$$
Therefore, $W^0_T(S^{\theta}, H, \lambda) \ge f^{\theta}, a.s.$ for all $\theta \in D$. Secondly, for every $\theta \in D$, the triangle inequality and  Theorem 1 of \cite{schachermayer2014admissible} again yield that 
\begin{eqnarray*}
	W^0_t(S^{\theta},\tilde{H}^{n},\lambda) &=& - \int_0^t{(1+\lambda)S^{\theta}_ud\tilde{H}^{n,\uparrow}_u} +  \int_0^t{(1-\lambda)S^{\theta}_ud\tilde{H}^{n,\downarrow}_u} \\
	&+& \tilde{H}^{n}_tS^{\theta}_t - \lambda |\tilde{H}^{n}_t|S^{\theta}_t\\ &\ge& -x^{\theta}, \text{ a.s. for } 0 \le t \le T. 
\end{eqnarray*}
In other words, $H \in \mathcal{A}^{\theta}_{x^{\theta}}(\lambda)$ for all $\theta \in D$.
\end{proof}
\begin{remark}
{\rm{}
Fatou-closedness is studied in the quasi-sure approach by \cite{maggis2018fatou}. In their topological setting, 
Fatou-closedness, denoted by $(FC)$, does not imply ``weak star'' closedness, denoted by $(WC)$. An additional condition, 
namely $\mathcal{P}$-sensitivity, is required. And it is proved that for a convex and monotone set, 
$(WC) = (FC) + \mathcal{P}$-sensitivity. When the set of priors $\mathcal{P}$ is dominated, $\mathcal{P}$-sensitivity 
is always satisfied, but this is not the case when $\mathcal{P}$ is non-dominated. In comparison, our Proposition 
\ref{pro:dense_closed} resembles the dominated case (although the laws of $S^{\theta}$ are not necessarily dominated), 
where the number of uncertain models is finite. Nevertheless, our ``weak star" closedness, that is 
$\bw^*$-closedness, is proved in Proposition \ref{pro:local_closed} below by using a different technique. Fatou-closedness 
is also discussed in the pathwise approach by \cite{dito2020} as a given property rather than a proved one.}    
\end{remark}
Recall
$$\bC = \left\lbrace (\bW^0_{T}(H, \lambda) - \bh) : H \in \A(\lambda), \bh \in \bL^0_+\right\rbrace \bigcap \prod_{\theta \in \Theta}L^{\infty}(\mathcal{F}_T,P).$$
Now, we are able to prove
\begin{proposition}\label{pro:local_closed}
The convex cone $\bC$ is $\bw^*$-closed.
\end{proposition}
\begin{proof}
Let $\bff_{\alpha}, \alpha \in I$ be a net in $\bC$, i.e. $ \bff_{\alpha} \le \bW^0_T(H_{\alpha}), H_{\alpha} \in \A(\lambda),$ 
such that $\bff_{\alpha} \to \bff$ in the $\bw^*$ topology for some $\bff \in \bL^{\infty}$. We need to prove that $\bff \in \bC$, that is there exists $H \in \mathcal{A}(\lambda)$ such that for all $\theta \in \Theta$, $W^0_T(S^{\theta}, H, \lambda) \ge f^{\theta}, a.s..$ 

For each $\theta \in \Theta$, we define 
$$\mathcal{H}^{\theta} = \{ H \in \bV:  H \in \A^{\theta}(\lambda) \text{ and } W^0_T(S^{\theta}, H, \lambda) \ge f^{\theta}, a.s. \}.$$
The set $\mathcal{H}^{\theta}$ is clearly convex for each $\theta \in \Theta$. In addition, using Lemma 4.3 of \cite{guasoni2002optimal}, we can prove that $\mathcal{H}^{\theta}$ is closed in $\mathbf{V}$. If we can prove 
\begin{equation}\label{eq:cap}
\bigcap_{\theta \in \Theta} \H^{\theta} \ne \emptyset
\end{equation}
then the proof is complete. First, we will prove that 
\begin{equation}\label{eq:finite_intersec}
\mathcal{H}^{D}=\bigcap_{\theta \in D} \H^{\theta} \ne \emptyset,
\end{equation}
where $D$ is an arbitrary finite subset of $\Theta$.

Proposition \ref{pro:dense_closed} implies that the set $\bC(D)$ 
is Fatou-closed. Therefore, using Proposition \ref{pro:weak_closed_ball}, the set $\bC(D)$ is closed in the $\bw^*$ topology of $\prod_{\theta \in D}L^{\infty}(\mathcal{F}_T,P)$. Since $f^{\theta}_{\alpha} \to f^{\theta}$ in the $w^*$ topology for each $\theta \in D$, and we obtain that $(f^{\theta})_{\theta \in D} \in \bC(D)$, and thus, (\ref{eq:finite_intersec}) holds true. 


Fix $\theta_0 \in \Theta$ arbitrarily. Since $x^{\theta_0}= \|f^{\theta_0}\| < \infty$, Lemma \ref{lem:bounded_P} shows that the set $\{ \|H\|_T, H \in \mathcal{A}^{\theta_0}_{x^{\theta_0}}(\lambda) \}$ is bounded in $L^0(\mathcal{F}_T,P)$, and thus convexly compact, by Proposition \ref{pro:convex_compact_V}. Since $$\bigcap_{\theta \in \Theta} \mathcal{H}^{\theta} = \bigcap_{D \in Fin(\Theta)} \mathcal{H}^{D \cup \{\theta_0\}} ,$$
we conclude that (\ref{eq:cap}) holds. The proof is complete.
\end{proof}

\begin{corollary}\label{cor}
For every $0 < \lambda < 1, \theta \in \Theta, A \in \mathcal{F}_T, P[A] >0$, the set $\mathcal{C}(\theta, \lambda, A)$ is nonempty. 
\end{corollary}
\begin{proof}
Let us fix $\bar\theta \in \Theta$ arbitrarily. By Proposition \ref{pro:local_closed}, the convex set $\bC$ is $\bw^*$-closed. The compact set $1_A\mathbf{1}^{\bar\theta}$ and the closed convex set $\bC$ are disjoint. Applying the Hahn-Banach theorem, there exists $\mathbf{Q} = (Z^{\theta}_T)_{\theta \in \Theta} \in \bigoplus_{\theta \in \Theta} L^1(\mathcal{F}_T,P)$ such that 
	$$  \sup_{\bff \in \bC} \mathbf{Q}(\bff) \le \alpha < \beta \le \mathbf{Q}(1_A\mathbf{1}^{\bar\theta}).$$
	Since $\mathbf{0} \in \bC$, it follows that $\alpha \ge 0$. Since $\bC$ is a cone, we must have 
	\begin{equation}\label{eq:Q}
	\mathbf{Q}(\mathbf{f}) \le 0,  \forall \mathbf{f} \in \bC,
	\end{equation}
and as a consequence,  $Z^{\theta}_T \ge 0, \forall \theta \in \Theta$. Note that $E[Z^{\bar\theta}_T1_A] > 0.$

For any stopping times $\sigma \le \tau \le T_n$ and $B \in \F_{\sigma}$, the strategy $H = \pm 1_{B} 1_{]\sigma, \tau] }$ belongs to $\mathcal{A}(\lambda)$. Therefore, from (\ref{eq:Q}) we obtain
\begin{equation}\label{eq:XY}
\bQ(((1-\lambda)\bS_{\tau} - (1+\lambda)\bS_{\sigma})1_B) \le 0, \qquad \bQ(((1-\lambda)\bS_{\sigma} - (1+\lambda) \bS_{\tau})1_B) \le 0.
\end{equation} 
Define 
$$X_{\tau} =  \sum_{\theta \in \Theta} Z^{\theta}_{\tau} (1-\lambda) S^{\theta}_{\tau}, \qquad  Y_{\sigma} = \sum_{\theta \in \Theta} Z^{\theta}_{\sigma} (1+ \lambda) S^{\theta}_{\sigma}.$$
where $Z^{\theta}_t = E[Z^{\theta}_{T}|\mathcal{F}_t]$.
We compute that
\begin{eqnarray*}
	&&E\left[ (X_{\tau} - Y_{\sigma}) 1_B \right] \\
	&=& E\left[ \left( \sum_{\theta \in \Theta}  Z^{\theta}_{\tau} (1-\lambda) S^{\theta}_{\tau}   -  \sum_{\theta \in \Theta}  Z^{\theta}_{\sigma} (1+ \lambda)S^{\theta}_{\sigma} \right) 1_B \right] \\
	&=& E\left[ \left( \sum_{\theta \in \Theta} E\left[  Z^{\theta}_{T}|\F_{\tau}\right]  (1-\lambda) S^{\theta}_{\tau}   \right) 1_B \right] - E\left[ \left( \sum_{\theta \in \Theta}    E\left[  Z^{\theta}_{T}|\F_{\sigma}\right] (1+\lambda) S^{\theta}_{\sigma} \right) 1_B \right] \\
	&\le& 0,
\end{eqnarray*}
by the tower law of conditional expectation and (\ref{eq:XY}). Similarly, we obtain
$$E\left[ (X_{\sigma} - Y_{\tau}) 1_B \right] \le 0. $$
Using Lemma \ref{lem:sandwich}, there is a martingale $M^n$ such that on $\llbracket0,T_n\rrbracket$,
$$(1-\lambda) \sum_{\theta \in \Theta} Z^{\theta}_t S^{\theta}_t \le M^n_t \le (1+ \lambda) \sum_{\theta \in \Theta} Z^{\theta}_t S^{\theta}_t, a.s.$$
However, from the proof of Lemma 6.2 of \cite{guasoni2012fundamental}, it can be checked that $M^{n+1}$ and $M^n$ coincide on $\llbracket0,T_n\rrbracket$. Therefore, the local martingale $M$ obtained by pasting the processes $M^n, n \in \mathbb{N}$ together satisfies
$$(1-\lambda_t) \sum_{\theta \in \Theta} Z^{\theta}_t S^{\theta}_t \le M_t \le  (1+\lambda_t)\sum_{\theta \in \Theta} Z^{\theta}_t S^{\theta}_t, a.s., t \in [0,T].$$
In other words, we obtain an element in $\mathcal{C}(\bar\theta,\lambda, A)$. 
\end{proof}

\begin{remark}
	One may ask: if Assumption \ref{assumption_main} (ii) is relaxed but a 
	robust no free lunch condition is imposed, is the set $\bC$ (or an appropriate enlargement of 
	it) $\bw^*$-closed? This would result in a robust version of FTAP but we do not know yet how to handle this case.
\end{remark}

We recall Lemma 6.3 of \cite{guasoni2012fundamental}.
\begin{lemma}\label{lem:sandwich}
	Let $(X_t)_{t \in [0,T]}$ and $(Y_t)_{t \in [0,T]}$ be two c\`adl\`ag bounded processes. The following conditions are equivalent:
	\begin{itemize}
		\item[(i)] There exists a c\`adl\`ag martingale $(M_t)_{t \in [0,T]}$ such that
		$$X \le M \le Y, \qquad a.s. $$
		\item[(ii)] For all stopping times $\sigma, \tau$ such that $0 \le \sigma \le \tau \le T, a.s.$, we have
		\begin{equation*}\label{sandwich2}
		E[X_{\tau}|\mathcal{F}_{\sigma}] \le Y_{\sigma}, \qquad \text{and} \qquad E[Y_{\tau}|\mathcal{F}_{\sigma}] \ge X_{\sigma}, \qquad a.s..
		\end{equation*}
	\end{itemize}
\end{lemma}

\subsection{Proof of Theorem \ref{super}} 

Let $z$ and $H \in \mathcal{A}(\lambda)$ satisfy
$$z + W^0_T(S^{\theta},H,\lambda) \ge G^{\theta}, a.s., \ \theta \in \Theta.$$
For any $(\bZ, M) \in \mathcal{Z}(\lambda)$, we have
\begin{eqnarray}\label{eq1}
z + E\left[ \sum_{\theta \in \Theta} Z^{\theta}_T W^0_T(S^{\theta}, H, \lambda) \right] \ge E\left[ \sum_{\theta \in \Theta} Z^{\theta}_T G^{\theta} \right] .
\end{eqnarray}
Suppose that $Z^{\theta}_T$ has positive values only if $\theta \in D$ for some $D \in Fin(\Theta)$. 
Denote by $H^{\uparrow,c}, H^{\downarrow,c}$ the continuous parts of $H^{\uparrow}, H^{\downarrow}$, respectively. We define the process $I^{\theta}_t := \int_0^t{(1+\lambda)S^{\theta}_u dH^{\uparrow}_u}$. Using integration by parts, we obtain that
\begin{eqnarray}
d\left( Z^{\theta}_tI^{\theta}_t\right) &=& I^{\theta}_{t-} dZ^{\theta}_t + Z^{\theta}_{t-} (1+\lambda)S^{\theta}_t dH^{\uparrow}_t  + d[Z^{\theta},I^{\theta}]_t \nonumber\\
&=& I^{\theta}_{t-} dZ^{\theta}_t + Z^{\theta}_{t-} (1+\lambda)S^{\theta}_t dH^{\uparrow,c}_t + Z^{\theta}_{t-} (1+\lambda)S^{\theta}_t \Delta H^{\uparrow}_t + \Delta Z^{\theta}_t \Delta I^{\theta}_t \nonumber\\
&=& I^{\theta}_{t-} dZ^{\theta}_t + Z^{\theta}_{t-} (1+\lambda)S^{\theta}_{t-} dH^{\uparrow,c}_t + Z^{\theta}_{t} (1+\lambda)S^{\theta}_t \Delta H^{\uparrow}_t \label{eq:p1},
\end{eqnarray}
noting that $S^{\theta}$ is continuous and $I^{\theta}$ is of finite variation. Similarly, we define $J^{\theta}_t = \int_0^t{(1-\lambda)S^{\theta}_u dH^{\downarrow}_u}$ and compute
\begin{equation}\label{eq:p2}
d\left( Z^{\theta}_tJ^{\theta}_t\right) = J^{\theta}_{t-} dZ^{\theta}_t + Z^{\theta}_{t-} (1-\lambda)S^{\theta}_{t-} dH^{\downarrow,c}_t + Z^{\theta}_{t} (1-\lambda)S^{\theta}_t \Delta H^{\downarrow}_t.
\end{equation}
Therefore, (\ref{eq:p1}), (\ref{eq:p2}) and the property of $M$ yield
\begin{eqnarray*}
	\sum_{\theta \in D} Z^{\theta}_tW^0_t(S^{\theta}, H, \lambda) &\le& \sum_{\theta \in D} \int_0^t{(J^{\theta}_{u-} - I^{\theta}_{u-}) dZ^{\theta}_u}  \\
	&-& \int_0^t{M_{u-}dH^{c}_u} + \sum_{u \le s \le t} M_u \Delta H_u + H_t M_t\\
	&=& \sum_{\theta \in D} \int_0^t{(J^{\theta}_{u-} - I^{\theta}_{u-}) dZ^{\theta}_u} + \int_{0}^t{H_{u-} dM_u}. 
\end{eqnarray*}
The RHS of the above inequality is a local martingale and thus a supermartingale as $W^0_t(S^{\theta}, H,\lambda)$ is uniformly bounded from below. From (\ref{eq1}), we have 
\begin{eqnarray*}
	E\left[ \sum_{\theta \in D} Z^{\theta}_TG^{\theta}  \right] &\le& z + E\left[ \sum_{\theta \in D} Z^{\theta}_T W^0_T(S^{\theta}, H, \lambda)  \right]\\
	&=& z.
\end{eqnarray*}
Therefore, $z \ge \sup_{(\bZ,M) \in \mathcal{Z}(\lambda)} E\left[ \sum_{\theta \in \Theta} Z^{\theta}_TG^{\theta}  \right].$

Next, we prove the reverse inequality for the case $G^{\theta}, \theta \in \Theta$ are bounded. Let $z \in \mathbb{R}$ be such that there is no strategy $H \in \mathcal{A}(\lambda)$ satisfying
$$ W^z_T(S^{\theta}, H, \lambda) \ge G^{\theta}, a.s., \ \forall \theta \in \Theta.$$
In other words, $(G^{\theta})_{\theta \in \Theta} - z \notin \bC.$ Applying the Hahn-Banach theorem, there exists $\mathbf{Q} = (Z^{\theta}_T)_{\theta \in \Theta} \in \bigoplus_{\theta \in \Theta} L^1(\mathcal{F}_T,P)$ such that 
$$  \sup_{\bff \in \bC} \mathbf{Q}(\bff) \le \alpha < \beta \le \mathbf{Q}\left((G^{\theta})_{\theta \in \Theta} - z \right) .$$ 
Since $\bC$ is a cone containing $-\bL^{\infty}_+$, it is necessarily that
$$ \sup_{\bff \in \bC} \mathbf{Q}(\bff)  = 0, \qquad  \mathbf{Q}\left((G^{\theta})_{\theta \in \Theta} - z \right) > 0.$$
We also deduce that $Z^{\theta}_T \ge 0, a.s., \theta \in \Theta$ and it is possible to normalize $\bQ$ such that $\bQ(\mathbf{1}) = 1$. A similar argument as in the proof of Corollary \ref{cor} gives the martingale $M$ associated to $\bZ$. This means $(\bZ, M) \in \mathcal{Z}(\lambda)$ and that 
$$z < \mathbf{Q}\left((G^{\theta})_{\theta \in \Theta} \right) \le \sup_{(\bZ,M) \in \mathcal{Z}(\lambda) } E\left[ \sum_{\theta \in \Theta} Z^{\theta}_T G^{\theta} \right] .$$

Finally we investigate the case $G^{\theta} \ge 0, \theta \in \Theta$. Let $z \in \mathbb{R}$ be the number such that
$z \ge \sup_{(\bZ,M) \in \mathcal{Z}(\lambda) } E\left[ \sum_{\theta \in \Theta} Z^{\theta}_T G^{\theta}) \right].$ Then, for all $n \in \mathbb{N}$, we have
$$z \ge  \sup_{(\bZ,M) \in \mathcal{Z}(\lambda) } E\left[ \sum_{\theta \in \Theta} Z^{\theta}_T G^{\theta} \wedge n \right]. $$
The result for bounded $G$ implies for each $n \in \mathbb{N}$, there exists $H^n \in \mathcal{A}(\lambda)$ such that
$$ z + W^0_T(S^{\theta},H^n, \lambda) \ge G^{\theta} \wedge n, a.s., \forall \theta \in \Theta.$$
For each $\theta \in \Theta$, Theorem 1 of \cite{schachermayer2014admissible} yields $H^n \in \mathcal{A}^{\theta}_{z}(\lambda)$ for all $n \in \mathbb{N}$. Now, we repeat the argument in Proposition \ref{pro:dense_closed}. For a fixed $\theta_0 \in \Theta$, the set $\{\|H\|_T, H \in \mathcal{A}^{\theta_0}_{z}(\lambda) \}$ is bounded in probability by Lemma 3.3. Using Lemma B.4 of \cite{guasoni2012fundamental}, there are convex combinations  $\tilde{H}^{n,\uparrow}, \tilde{H}^{n,\downarrow}$ and finite variation processes $H^{\uparrow}, H^{\downarrow}$ such that $\tilde{H}^{n,\uparrow} \to H^{\uparrow}, \tilde{H}^{n,\downarrow} \to H^{\downarrow}$, and $\|\tilde{H}^{n,\uparrow}\| \to \|H^{\uparrow}\|,\|\tilde{H}^{n,\downarrow}\| \to \|H^{\downarrow}\|$ pointwise. From Lemma 4.3 of \cite{guasoni2002optimal}, we obtain 
$\tilde{H}^{n,\uparrow} \to H^{\uparrow}$ and  $\tilde{H}^{n,\downarrow} \to H^{\downarrow}$ weakly. It follows that for every $\theta \in \Theta, t \in [0,T]$,
$$W^0_t(S^{\theta},\tilde{H}^{n}, \lambda) \to W^0_t(S^{\theta},H, \lambda), a.s..$$
Since $\mathbf{G}$ is bounded from below, $H \in \mathcal{A}(\lambda)$ and  $z +  W^0_T(S^{\theta},H, \lambda) \ge G^{\theta}, a.s., \theta \in \Theta.$ The proof is complete.
\section{Auxiliary results}\label{sec:app}






\subsection{Direct sums and product spaces}

A bilinear pairing is a triple $(X,Y,\left\langle\cdot , \cdot\right\rangle )$ where $X,Y$ are vector 
spaces over $\mathbb{R}$ and $\left\langle \cdot,\cdot \right\rangle$ is a bilinear map from $X \times Y$ to $\mathbb{R}$. 
Let $(E,u)$ be a topological vector space. Let $E^*=(E,u)^*$ be its dual, i.e.,\ the set of all 
continuous linear maps from $E$ to $\mathbb{R}$. Then there is a natural bilinear
pairing $(E,E^*, \left\langle\cdot,\cdot\right\rangle)$. We denote by $\sigma(E,E^{*})$ the usual weak topology
on $E$ and by $\sigma(E^{*},E)$ the weak-star topology on $E^{*}$.


Let $I$ be a non-empty set and, for each $i \in I$, let $(X_i,\tau_i)$ be a locally convex topological spaces. 
The topological direct sum of the family $(X_i,\tau_i)$, denoted by 
$\bigoplus_{i \in I} (X_i, \tau_i)$, is the locally convex space defined as follows. 
The vector space $\bigoplus_{i \in I} X_i$ is the set of tuples $(x_i)_{i \in I}$ with $x_i \in X_i$ such that 
$x_i =0$ for all but finitely many $i$. It is equipped with the inductive topology with respect to the canonical embeddings 
\begin{eqnarray*}
e_i: (X_i,\tau_i) &\to& X\\
x_i &\mapsto& x = (x^i),
\end{eqnarray*}
where $x^i = x_i$ and $x^j = 0$ whenever $j \ne i$, i.e.\ the 
strongest locally convex topology on $\bigoplus_{i \in I} X_i$ such that all these embeddings are continuous. 

The product space of  of the family $(X_i,\tau_i)$, denoted by $\prod_{i \in I}(X_i,\tau_i)$, 
consists of the product set $\prod_{i \in I}X_i$ and a topology $\tau$ having as its basis the family
$$\left\lbrace  \prod_{i \in I}O_i:  O_i \in \tau_i 
\text{ and } O_i = X_i \text{ for all but a finite number of } i \right\rbrace .$$
The topology $\tau$ is called the product topology, which is the coarsest topology for which all 
the projections are continuous. Note that the product space defined in 
this way is also a topological vector space, see Theorem 5.2 of \cite{guide2006infinite}. 
Since each $(X_i, \tau_i)$ is locally convex, $\prod_{i \in I}(X_i,\tau_i)$ is locally convex, too, 
see Proposition 2.1.3 of \cite{bogachev2017topological}. 
If $I$ is uncountable, the product space is not normable.  

For any index set $I$, it holds that
\begin{equation}\label{eq:duality_sum_prod}
\left( \bigoplus_{i \in I} (X_i, \tau_i) \right)^* = \prod_{i \in I} X^*_i, \qquad 
\left( \prod_{i \in I} (X_i,\tau_i) \right)^* = \bigoplus_{i \in I} X^*_i,
\end{equation}
see Corollary 1, page 138 and Theorem 4.3, page 137 of \cite{schaefer1971locally}.

We will be using the pairing
$$\left\langle \bff,\mathbf{g} \right\rangle = 
\sum_{i \in I} \left\langle f^i,g^i\right\rangle_{i} , 
\qquad \forall \bff \in \bigoplus_{i \in I} X_i, \mathbf{g} \in \prod_{i \in I} X^*_i,$$
where $\langle\cdot,\cdot\rangle_{i}$ is the natural pairing for $X_{i},X_{i}^{*}$. From Corollary 1, page 138 of \cite{schaefer1971locally}, it holds that
\begin{equation}\label{eq:dual2}
\sigma\left(\prod_{i \in I} X^*_i, \bigoplus_{i \in I} X_i \right) = \prod_{i \in I} \sigma(X^*_i, X_i).
\end{equation}


Let $D$ be a finite index set. In what follows, we will be interested in the duality between 
\begin{equation}\label{eq:EE*}
\mathbf{E} := \bigoplus_{\gamma \in D} 
(L^1(\mathcal{F}_T,P),\|\cdot\|_1), \qquad \mathbf{E}^* =\prod_{\gamma \in D} L^{\infty}(\mathcal{F}_T,P).
\end{equation}
We define $B^{\infty}_{r}= \{ f \in L^{\infty}(\mathcal{F}_T,P): \|f\|_{\infty} \le r\}$, the closed ball of radius $r\geq 0$ in 
$L^{\infty}(\mathcal{F}_T,P)$. The following result is analogous to Proposition 5.2.4 of \cite{delbaen2006mathematics}.  
\begin{proposition}\label{pro:weak_closed_ball}
Let $\bC \subset \mathbf{E}^*$ be a convex set, where $\mathbf{E}^*$ is defined in (\ref{eq:EE*}). 
The set $\bC$ is closed in the $\bw^*$ topology if and only if 
$\bC \cap \prod_{\gamma \in D} B^{\infty}_{r}$ is closed in $\prod_{\gamma \in D} L^{0}(\mathcal{F}_T,P)$
for each $r\geq 0$.
\end{proposition}
\begin{proof} We follow the proof of Proposition 4.4 in \cite{kabanov-last}.
$(\Rightarrow)$ 
Let $\bff_n, n \in \mathbb{N}$ be a sequence in  
$\bC \cap \prod_{\gamma \in \Gamma} B^{\infty}_{r}$ such that 
$\bff_n \to \bff$ in $\prod_{\gamma \in D} L^{0}(\mathcal{F}_T,P)$. 
We have to show that $\bff \in \bC \cap \prod_{\gamma \in \Gamma} B^{\infty}_{r}.$ 
For each $g \in L^1(\mathcal{F}_T,P)$ and $\gamma\in D$, the dominated convergence theorem 
implies that 
$\lim_{n \to \infty} E[gf^{\gamma}_n] = E[gf^{\gamma}]$. 
Therefore we have that $\bff \in \bC \cap \prod_{\gamma \in \Gamma} B^{\infty}_{r}.$

$(\Leftarrow)$ Assume that $\bC \cap \prod_{\gamma \in D} B^{\infty}_{r}$ is 
closed in $\prod_{\gamma \in D} L^{0}(\mathcal{F}_T,P)$. 
It is easy to check that $\bC \cap \prod_{\gamma \in D} B^{\infty}_{r}$ is 
closed in the Hilbert space $\prod_{\gamma \in D} (L^{2}, \|\cdot\|_2)$ hence also
in its weak topology $\sigma(\prod_{\gamma \in D} L^{2},\prod_{\gamma \in D} L^{2})$,{}
which is the same as $\sigma(\prod_{\gamma \in D} L^{\infty},\prod_{\gamma \in D} L^{2})$-closedness.
Since $L^{2}\subset L^{1}$, the set $\bC \cap \prod_{\gamma \in D} B^{\infty}_{r}$ is also
$\sigma(\prod_{\gamma \in D} L^{\infty},\prod_{\gamma \in D} L^{1})$-closed so,
by the Krein-{S}mulian theorem, $\bC$ is closed in the weak-star topology.

\end{proof}

\subsection{Convex compactness in $(L^0_+)^{\mathbb{N}}$}\label{sec:convex_comp}


Let $L^{\S}$ denote the set of 
$[0,\infty]$-valued random variables, equipped with the topology of convergence in probability.
A set $A\subset L^0_{+}$ is \emph{bounded} if $\sup_{X\in A}P(X\geq n)\to 0$, $n\to\infty$.
Now consider the topological product $\mathbf{L}:=(L^0_+)^{\mathbb{N}}$.
We call a subset $C\subset\mathbf{L}$ \emph{c-bounded}, if $\pi_k(C)$
is bounded in $L^0_{+}$ for all coordinate mappings $\pi_k$, $k\in\mathbb{N}$.

For any set $A$ we denote by $\mathrm{Fin}(A)$ the family of all non-empty finite subsets of $A$. This
is a directed set with respect to the partial order induced by inclusion. 
We reproduce Definition 2.1 of \cite{gordan}.

\begin{definition}
	{\rm A convex subset $C$ of some topological vector space is \emph{convexly compact},
		if for any non-empty set $A$ and any family $F_a$, $a\in A$ of closed and convex
		subsets of $C$, one has $\cap_{a\in A}F_a\neq\emptyset$ whenever
		\[
		\forall B\in\mathrm{Fin}(A),\ \cap_{a\in B}F_a\neq\emptyset.
		\]}
\end{definition}

It was established, independently in both \cite{pratelli} and \cite{gordan}, that every
closed and bounded convex subset of $L_+^0$ is convexly compact. In this section we will show the following.

\begin{proposition}\label{bfl}
	Any c-bounded, convex and closed subset $C\subset\mathbf{L}$ is convexly compact.
\end{proposition}

For an element $f\in\mathbf{L}$ we write $f^k:=\pi_k(f)$, $k\in\mathbb{N}$.
Let $I$ be a directed set and $f_i$, $i\in I$ a net in $\mathbf{L}$. For each $i\in I$, let $\Gamma_i$ denote the
set of (finite) convex combinations of the elements $\{f_j:\ j\geq i\}$. The next lemma goes back to Lemma A1.1
of \cite{delbaen1994general}. Its proof closely follows that of Lemma 2.1 in 
\cite{pratelli}, see also Theorem 3.1 in \cite{gordan}.

\begin{lemma}\label{maj} There exist $g_i\in \Gamma_i$, $i\in I$ such that the nets $g^k_i$, $i\in I$ converge to $g^k$ in probability, for
	each $k\in\mathbb{N}$, where $g^k\in L^{\S}$. 
\end{lemma} 
\begin{proof} Set $u(x):=1-e^{-x}$, $x\in [0,\infty]$ and note that, for given $\alpha>0$, there is $\beta>0$ such that
	\begin{equation}\label{beta}
	u\left(\frac{x+y}{2}\right)\geq \frac{u(x)+u(y)}{2}+\beta,\mbox{ when }|x-y|\geq\alpha\mbox{ and }\min(x,y)\leq 1/\alpha.
	\end{equation}
	Set, for $i\in I$,
	\[
	s_i=\sup\left\{ \sum_{k=0}^{\infty} 2^{-k}Eu(g^k):\ g\in\Gamma_i\right\}.
	\]
	As $s_i$, $i\in I$ is a non-increasing net of numbers, it converges to $s_{\infty}:=\inf_{i\in I} s_i$.
	Choose a non-decreasing $i_m$, $m\in\mathbb{N}$ such that $s_{\infty}=\lim_{m\to\infty} s_{i_m}$
	and, for each $m$,
	\begin{equation}\label{bech}
	|s_{\infty}-s_{i_m}|\leq \frac{1}{m+1}
	\end{equation}
	and let $g_{i_m}$ be such that $\sum_{k=0}^{\infty} 2^{-k}Eu(g^k_{i_m})\geq s_{i_m}-1/(m+1)$.
	
	For elements $p\in I\setminus\{i_m, m\in\mathbb{N}\}$, there is $l=l(p)$ such that $s_{i_{l+1}}\leq s_p\leq s_{i_l}$ and choose
	$g_p\in\Gamma_p$ such that $\sum_{k=0}^{\infty}2^{-k}Eu(g^k_p)\geq s_{i_{l+1}}-1/(l+1)$.
	
	In order to prove that $g^k_i$, $i\in I$ is a Cauchy net in $L^{\S}$
	for each $k$, we need to establish that, for each $k$ and each $\alpha,\epsilon>0$,
	there is $i(\epsilon)$ such that, for $p,q\geq i(\epsilon)$,
	\[
	P\left(|g_p^k-g^k_q|\geq\alpha,\ \min(g_p^k,g_q^k)\leq 1/\alpha\right)\leq \epsilon.
	\]
	Let $p,q\geq i_m$. Notice that $(g_p+g_q)/2\in\Gamma_{i_m}$ so \[
	\sum_{k=0}^{\infty}2^{-k}Eu\left(\frac{g_p^k+g_q^k}{2}\right)\leq s_{i_m}.
	\]
	However, by construction, $l(p)\geq m$ so
	\[
	\sum_{k=0}^{\infty}2^{-k}Eu\left(g_p^k\right)\geq s_{i_{l(p)+1}}-\frac{1}{l(p)+1}\geq s_{i_m}-\frac{1}{m+1}-(s_{i_m}-s_{i_{l(p)+1}}),
	\]
	and the latter is $\geq s_{i_m}-\frac{2}{m+1}$, by \eqref{bech}. A similar estimate holds
	for $\sum_{k=0}^{\infty}2^{-k}Eu\left(g_q^k\right)$ hence we conclude, from \eqref{beta}, that
	\begin{eqnarray*}
		\beta\sum_{k=0}^{\infty} 2^{-k}P\left(|g_p^k-g^k_q|\geq\alpha,\ \min(g_p^k,g_q^k)\leq 1/\alpha\right) &\leq&\\
		s_{i_m}-\frac{s_{i_m}}{2}+\frac{1}{m+1}-\frac{s_{i_m}}{2}+\frac{1}{m+1} &\leq& \frac{2}{m+1}.
	\end{eqnarray*}
	This entails, for each $k\in\mathbb{N}$, 
	\[
	P\left(|g_p^k-g^k_q|\geq\alpha,\ \min(g_p^k,g_q^k)\leq 1/\alpha\right)\leq \frac{2^{k+1}}{\beta(m+1)},
	\]
	which can be made arbitrarily small if $m$ is large enough. The statement is proved.
\end{proof}

\begin{proof}[Proof of Proposition \ref{bfl}.]
	Let $A$ be an arbitrary index set and let $C_a$, $a\in A$ be closed, convex subsets of $C$. Assume that for each $\mathbf{a}\in\mathrm{Fin}(A)$ with 
	$\mathbf{a}=\{a_1,\ldots,a_K\}$ we have $C_{a_1}\cap\ldots\cap C_{a_K}\neq \emptyset$. Let us pick an element 
	$c(\mathbf{a})$ from this
	intersection. Apply Lemma \ref{maj} to the net $c(\mathbf{a})$, $\mathbf{a}\in \mathrm{Fin}(A)$ 
	to obtain convex combinations $g(\mathbf{a})\in C_{a_1}\cap\ldots\cap C_{a_K}$
	such that the net $g^k(\mathbf{a})$, $\mathbf{a}\in\mathrm{Fin}(A)$ converges to some $g^k\in L^{\S}$ in probability, for each $k\in\mathbb{N}$.
	
	As $C$ is c-bounded, $g^k(\mathbf{a})$, $\mathbf{a}\in A$ are bounded in $L^0_{+}$, so, actually, $g^k\in L^0_{+}$. It follows that $g\in\mathbf{L}$ and,
	by the definition of topological products, the net $g(\mathbf{a})$, $\mathbf{a}\in\mathrm{Fin}(A)$ converges to $g$.
	However, for each fixed $a\in A$, $g(\mathbf{a})\in C_a$ for each $\mathbf{a}\in\mathrm{Fin}(A)$ containing $a$, hence, by the closedness of $C_a$,
	we also have $g\in C_a$. It follows that $\cap_{a\in A}C_a\neq\emptyset$ since $g$ is in this intersection.
\end{proof}

\subsection{Convex compactness for finite variation processes}\label{section:v_space}

Let $\mathcal{V}$ denote the family of non-decreasing, left-continuous functions on $[0,T]$ which are $0$ at $0$.
Let $r_k$, $k\in\mathbb{N}$ be an enumeration of $\left(\mathbb{Q}\cap [0,T]\right)\cup \{T\}$ with $r_0=T$.
For $f,g\in\mathcal{V}$, define
\[
\rho(f,g):=\sum_{k=0}^{\infty} 2^{-k}|f(r_k)-g(r_k)|.
\]
The series converges since $|f(r_k)-g(r_k)|\leq f(T)+g(T)$ for each $k$, and it defines a metric. The corresponding Borel-field
is denoted by $\mathcal{G}$.

Let $\mathbf{V}$ denote the set of pairs 
$H=(H^{\uparrow},H^{\downarrow})$ where
$H^{\uparrow}_t,H^{\downarrow}_t$, $t\in [0,T]$ are predictable processes 
such that $H^{\uparrow}(\omega),H^{\downarrow}(\omega)\in\mathcal{V}$ for each 
$\omega\in\Omega$.
Considered as mappings $H^{\uparrow},H^{\downarrow}:(\Omega, \mathcal{F}) \to (\mathcal{V}, \mathcal{G})$, they 
are measurable, 
by the definition of the metric $\rho$.
{}
We identify elements of $\mathbf{V}$ when they coincide outside a $P$-null set.
We equip $\mathbf{V}$ with
the topology coming from the metric 
\[
\varrho(H,G):=E[\rho(H^{\uparrow},G^{\uparrow})\wedge 1]+E[\rho(H^{\downarrow},G^{\downarrow})\wedge 1].
\]
Although this metric wasn't defined, a
related convergence structure was introduced already in \cite{chau2019robust}.

Similarly to Proposition \ref{bfl}, we obtain the following 
convex compactness result for subsets of $\mathbf{V}$.

\begin{proposition}\label{pro:convex_compact_V}
Let $C$ be a convex and closed subset of $\mathbf{V}$. 
If $$\{H^{\uparrow}_T + H^{\downarrow}_T, H \in C \}$$ is bounded in $L^{0}_{+}$ then $C$ is convexly compact.
\end{proposition} 
\begin{proof}
Let $A$ be an arbitrary index set and let $C_a$, $a\in A$ be closed, convex subsets of $C$. Assume that for any $\mathbf{a}\in\mathrm{Fin}(A)$ with 
$\mathbf{a}=\{a_1,\ldots,a_K\}$ we have $C_{a_1}\cap\ldots C_{a_K}\neq \emptyset$. We prove that $\cap_{a\in A}C_a\neq\emptyset$.
We identify $\mathbf{L}$ with $(L_{+}^{0})^{(\mathbb{Q}\cap [0,T])\cup \{T\}}$.

Denote 
$$D = \left\lbrace (H^{\uparrow}_q,H^{\downarrow}_q)_{q \in \left(\mathbb{Q}\cap [0,T]\right)\cup \{T\} }, 
H \in C \right\rbrace\subset \mathbf{L} $$
and similarly
$$D_{a} = \left\lbrace (H^{\uparrow}_q,H^{\downarrow}_q)_{q \in \left(\mathbb{Q}\cap [0,T]\right)\cup \{T\} }, H \in C_{a} \right\rbrace  \subset D.$$
Clearly, $D_{a}, D$ are convex and closed in $\mathbf{L}$. By hypothesis, the set $D$ is $c$-bounded. 
Therefore, Proposition \ref{bfl} implies there exists 
$(\bar{H}_q^{\uparrow},\bar{H}_q^{\downarrow})_{q \in \left(\mathbb{Q}\cap [0,T]\right)\cup \{T\} }$ in $\cap_{a\in A}D_a\neq\emptyset$. Define, for $t\in [0,T]\setminus{}
\mathbb{Q}$, 
$$H^{\uparrow}_t = \lim_{q \uparrow t, q \in \mathbb{Q}} H^{\uparrow}_q, \qquad H^{\downarrow}_t = \lim_{q \uparrow t, q \in \mathbb{Q}} H^{\downarrow}_q.$$
Then $(\bar{H}^{\uparrow},\bar{H}^{\downarrow}) \in \cap_{a\in A}C_a$ as required. 
\end{proof}

\begin{corollary}\label{cor:ccv}
Let $C$ be a convex and closed subset of $\mathbf{V}$. 
If $\{H^{\uparrow}_T + H^{\downarrow}_T, H \in C \}$ is bounded in probability then,
for each $a>0$, $[-a,a]\times C$ is convexly compact (as a subset of $\mathbb{R}\times\mathbf{V}$).\hfill $\Box$
\end{corollary} 
\begin{proof}
Follows by obvious modifications of the arguments of Subsection \ref{sec:convex_comp} and Proposition \ref{pro:convex_compact_V}.	
\end{proof}


\bibliographystyle{plain}

\end{document}